\documentclass[11pt,a4paper]{amsart} 
\usepackage{amsthm,amsfonts,amsmath,calc,hyperref,enumerate,graphicx}
\usepackage[sort&compress,numbers]{natbib}
\usepackage[lmargin=35mm,rmargin=35mm,tmargin=30mm,bmargin=30mm,footskip=10mm]{geometry}

\theoremstyle{plain} 
\newtheorem{theorem}{Theorem}[section]
\newtheorem{lemma}[theorem]{Lemma}

\newtheorem{proposition}[theorem]{Proposition}

\theoremstyle{definition}

\newcommand{\thmlabel}[1]{\label{thm:#1}}
\newcommand{\thmref}[1]{Theorem~\ref{thm:#1}}

\newcommand{\threethmref}[3]{Theorems~\ref{thm:#1}, \ref{thm:#2} and \ref{thm:#3}}
\newcommand{\multithmref}[2]{Theorems~\ref{thm:#1}--\ref{thm:#2}}

\newcommand{\lemlabel}[1]{\label{lem:#1}}
\newcommand{\lemref}[1]{Lemma~\ref{lem:#1}}

\newcommand{\figlabel}[1]{\label{fig:#1}}
\newcommand{\figref}[1]{Figure~\ref{fig:#1}}
\newcommand{\Figref}[1]{Figure~\ref{fig:#1}}

\newcommand{\seclabel}[1]{\label{sec:#1}}
\newcommand{\secref}[1]{Section~\ref{sec:#1}}

\newcommand{\proplabel}[1]{\label{prop:#1}}
\newcommand{\propref}[1]{Proposition~\ref{prop:#1}}

\newcommand{\Figure}[4][htb]{
  \begin{figure}[#1]
    \vspace*{1ex}
    \begin{center}#3\end{center} 
    \vspace*{-1ex}
    \caption{\figlabel{#2}#4}
    \vspace*{1ex}
  \end{figure}
}

\newcommand{\third}{\ensuremath{\protect\tfrac{1}{3}}}

\newcommand{\half}{\ensuremath{\protect\tfrac{1}{2}}}

\DeclareMathOperator{\outdeg}{outdeg}
\DeclareMathOperator{\indeg}{indeg} 

\newcommand{\RNG}{{\sf RNG}}
\newcommand{\MST}{{\sf MST}} \newcommand{\LUNE}{{\sf lune}}
\newcommand{\DISC}{{\sf disc}} \newcommand{\CDISC}{{\overline{\DISC}}}
\newcommand{\LENS}{{\sf lens}} \newcommand{\CIRC}{{\sf circle}}

\newcommand{\RR}{\mathbb{R}^2}

\newcommand{\PP}{\ensuremath{\mathcal{P}}}

\renewcommand{\baselinestretch}{1.25}
\begin{document}

\title{Proximity Drawings of High-Degree Trees}

\thanks{The research of Ferran Hurtado is partially supported by
  projects MTM2009-07242 and Gen.\ Cat 2009SGR1040. The research of
  Giuseppe Liotta is supported by CNR and MURST. David Wood is
  supported by a QEII Research Fellowship from the Australian Research
  Council; research initiated at Universitat Polit{\`e}cnica de
  Catalunya, where supported by the Marie Curie Fellowship
  MEIF-CT-2006-023865, and by the projects MEC MTM2006-01267 and DURSI
  2005SGR00692.}

\author{Ferran Hurtado} \address{\newline Departament de
  Matem{\`a}tica Aplicada II\newline Universitat Polit{\`e}cnica de
  Catalunya\newline Barcelona, Spain}  \email{ferran.hurtado@upc.edu}
 
\author{Giuseppe Liotta} \address{\newline Dipartimento di
  Ingegneria Elettronica e dell'Informazione\newline Universit\`{a} di
  Perugia\newline Perugia, Italy} \email{liotta@diei.unipg.it}

\author{David R. Wood} \address{\newline Department of
  Mathematics and Statistics\newline The University of
  Melbourne\newline Melbourne, Australia} \email{woodd@unimelb.edu.au}

\keywords{graph, tree, proximity graph, minimum spanning tree,
  relative neighbourhood graph, thickness}

\begin{abstract}
  A drawing of a given (abstract) tree that is a minimum spanning tree
  of the vertex set is considered aesthetically pleasing. However,
  such a drawing can only exist if the tree has maximum degree at most
  6. What can be said for trees of higher degree? We approach this
  question by supposing that a partition or covering of the tree by
  subtrees of bounded degree is given. Then we show that if the
  partition or covering satisfies some natural properties, then there
  is a drawing of the entire tree such that each of the given subtrees
  is drawn as a minimum spanning tree of its vertex set.
\end{abstract}

\maketitle

\section{Introduction}
\seclabel{Intro}

The field of graph drawing studies aesthetically pleasing drawings of
graphs\footnote{We consider graphs $G$ that are simple and finite. Let
  $G$ be an (undirected) graph. The \emph{degree} of a vertex $v$ of
  $G$, denoted by $\deg_G(v)$, is the number of edges of $G$ incident
  with $v$. The minimum and maximum degrees of $G$ are respectively
  denoted by $\delta(G)$ and $\Delta(G)$. We say $G$ is
  \emph{degree-$d$} if $\Delta(G)\leq d$. Now let $G$ be a directed
  graph. Let $v$ be a vertex of $G$. The \emph{indegree} of $v$,
  denoted by $\indeg_G(v)$, is the number of incoming edges incident
  to $v$. The \emph{outdegree} of $v$, denoted by $\outdeg_G(v)$, is
  the number of outgoing edges incident to $v$. The maximum outdegree
  of  $G$ is denoted by $\Delta^+(G)$.  We say $G$ is
  \emph{outdegree-$d$} if $\Delta^+(G)\leq d$.}. There are a number of
recognised criteria for measuring the quality of a drawing of a given
graph. These include:
\begin{itemize}
\item no two edges should cross in drawings of planar graphs;
\item the edges should be drawn as straight line-segments; and
\item the drawing should have large \emph{angular resolution} (defined
  to be the minimum angle determined by two consecutive edges incident
  to a vertex).
\end{itemize}
These three criteria are adopted in the present paper. More formally,
a (\emph{straight-line general position}) \emph{drawing} of graph $G$
is an injective function $\phi:V(G)\rightarrow\mathbb{R}^2$ such that
the points $\phi(u),\phi(v),\phi(w)$ are not collinear for all
distinct vertices $u,v,w\in V(G)$. The \emph{image} of an edge $vw\in
E(G)$ under $\phi$ is the line segment
$\overline{\phi(v)\phi(w)}$. Where no confusion is caused, we
henceforth do not distinguish between a graph element and its image in
a drawing. Two edges \emph{cross} if they intersect at a point other
than a common endpoint. 

Our focus is on drawings of trees. Here a number of other criteria
have been studied that will not be considered in this paper. These
include: small bounding box area \citep{CDTT-SJC95, Chan-Algo02,
  Kim-IPL97, SKC-CGTA00, CGKT-CGTA02, GargRusu-JGAA04, CDP-CGTA92},
small aspect ratio \citep{GargRusu-JGAA04,CGKT-CGTA02}, few bends in
the edges \citep{Kim-DAM04}, few distinct edge-slopes
\citep{DESW-CGTA}, few distinct edge-lengths \citep{CDMW-EJC08},
layered vertices \citep{Suderman-IJCGA04}, upwardness in rooted trees
\citep{Trevisan-IPL96, CDP-CGTA92, Kim-DAM04, Chan-Algo02}, and
maximising symmetry \citep{HE-Algo03}.

A \emph{minimum spanning tree} of a finite set $P\subset\mathbb{R}^2$,
denoted by $\MST(P)$, is a straight-line drawing of a tree with vertex
set $P$ and with minimum total edge length; see \Figref{MST} for an
example. A drawing of a given (abstract) tree that is a minimum
spanning tree of its vertex set is considered to be particularly
aesthetically pleasing. In particular, every minimum spanning tree is
crossing-free and has angular resolution at least
$\frac{\pi}{3}$. Drawings defined in this way are called `proximity
drawings'; see \secref{TreeDrawing} and
\citep{ADH-Delauney,ADH-Gabriel,PV-CGTA04,MonmaSuri-DCG92,BLL-Algo96,DLW-JDA06,Liotta-ProximityDrawings}
for more on proximity drawings.

\Figure{MST}{\includegraphics{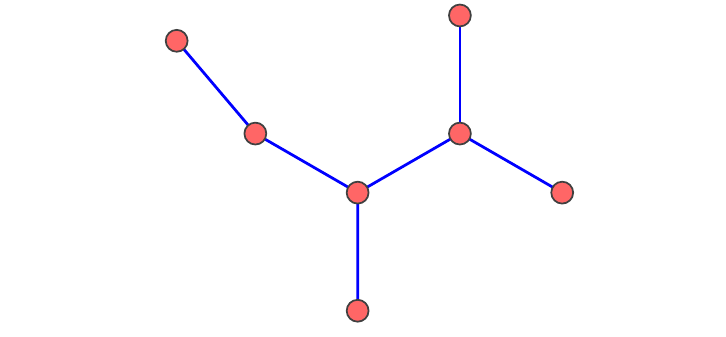}}{Example of a minimum spanning
  tree.}

\citet{MonmaSuri-DCG92} proved that every degree-$5$ tree can be drawn
as a minimum spanning tree of its vertex set, and they provided a
linear time (real RAM) algorithm to compute the drawing. In any
drawing of a vertex $v$ with degree at least $7$, some angle at $v$ is
greater than $\frac{\pi}{3}$, and the same is true for a degree-$6$
vertex if the points are required to be in general position. Thus a
tree that contains a vertex with degree at least $7$ cannot be drawn
as a minimum spanning tree, and the same is true for a degree-$6$
vertex if the points are in general position. If collinear vertices
are allowed, then \citet{EadesWhitesides-Algo96} showed that it is
NP-hard to decide whether a given degree-$6$ tree can be drawn as a
minimum spanning tree. In this sense, the problem of testing whether a
tree can be drawn as a minimum spanning tree is essentially
solved. (In related work, \citet{LM-CGTA03} characterised those trees
that have drawings that are Voronoi diagrams of their vertex set.)\

What can be said about drawings of a high degree tree $T$ that
`approximate' the minimum spanning tree of the vertex set? We prove
the following solutions to this question based on partitions of $T$
into subtrees of bounded degree. A \emph{partition} of a graph $G$ is
a set of subgraphs of $G$ such that every edge of $G$ is in exactly
one subgraph. A partition can also be thought of as a (non-proper)
edge-colouring, with one colour for each subgraph. We emphasise that
`trees' and `subtrees' are necessarily connected.

\begin{theorem}
  \thmlabel{Degree5Partition-MST} Let \PP\ be a partition of a tree
  $T$ into degree-$5$ subtrees. Then there is a drawing of $T$ such
  that each subtree in \PP\ is drawn as the minimum spanning tree of
  its vertex set.
\end{theorem}

The drawing of $T$ produced by \thmref{Degree5Partition-MST} possibly
has crossings, which are undesirable. The next result eliminates the
crossings, at the expense of a slightly stronger assumption about the
partition, which is expressed in terms of rooted trees. A \emph{rooted
  tree} is a directed tree such that exactly one vertex, called the
\emph{root}, has indegree $0$. It follows that every vertex except $r$
has indegree $1$, and every edge $vw$ of $T$ is oriented `away' from
$r$; that is, if $v$ is closer to $r$ than $w$, then $vw$ is directed
from $v$ to $w$. If $r$ is a vertex of a tree $T$, then the pair
$(T,r)$ denotes the rooted tree obtained by orienting every edge of
$T$ away from $r$.

\begin{theorem}
  \thmlabel{Degree4Partition-MST} Let \PP\ be a partition of a rooted
  tree $T$ into outdegree-$4$ subtrees. Then there is a non-crossing
  drawing of $T$ such that each subtree in \PP\ is drawn as the
  minimum spanning tree of its vertex set.
\end{theorem}

By further restricting the partition we introduce large angular
resolution as an additional property of the drawing, again at the
expense of a slightly stronger assumption about the partition.

\begin{theorem}
  \thmlabel{Degree3Partition-MST} Let \PP\ be a partition of a rooted
  tree $T$ into outdegree-$3$ subtrees. Then there is a non-crossing
  drawing of $T$ with angular resolution at least \linebreak
  $\frac{\pi}{\max\{\Delta^+(T)-1,4\}}$ such that each subtree in
  \PP\ is drawn as the minimum spanning tree of its vertex set.
\end{theorem}

Since every drawing of $T$ has angular resolution at most
$\frac{2\pi}{\Delta(T)}$, the bound on the angular resolution in
\thmref{Degree3Partition-MST} is within a constant factor of optimal.



Our final drawing theorem concerns a given covering of a tree by two
bounded degree subtrees. A \emph{covering} of a graph $G$ is a set of
connected subgraphs of $G$ such that every edge of $G$ is in at least
one subgraph.

\begin{theorem}
  \thmlabel{TwoCovering-MST} Let $\{T_1,T_2\}$ be a covering of a tree
  $T$ by two degree-$5$ subtrees. Then there is a non-crossing drawing
  of $T$ such that each $T_i$ is drawn as a minimum spanning tree of
  its vertex set.
\end{theorem}

A number of notes about
\multithmref{Degree5Partition-MST}{TwoCovering-MST} are in order:
\begin{itemize}
\item Each of
  \threethmref{Degree5Partition-MST}{Degree4Partition-MST}{TwoCovering-MST}
  imply and generalise the above-mentioned result by
  \citet{MonmaSuri-DCG92} that every degree-$5$ tree $T$ can be drawn
  as a minimum spanning tree of its vertex set. (Take $k=1$ in
  \thmref{Degree5Partition-MST}; root $T$ at a leaf in
  \thmref{Degree4Partition-MST}; and take $T_1=T$ and $T_2=\emptyset$
  in \thmref{TwoCovering-MST}.)\

\item \thmref{TwoCovering-MST} cannot be generalised for coverings by three or
more subtrees; see \secref{Covering}. 

\item The above theorems are loosely related to the notion of
  geometric thickness. The \emph{geometric thickness} of a graph $G$
  is the minimum integer $k$ such that there is a straight-line
  drawing of $G$ and an edge $k$-colouring such that monochromatic
  edges do not cross; see \citep{Eppstein-AMS, BMW-EJC06, DEH-JGAA00,
    Duncan,DujWoo-DCG07, DEK-SoCG04, Eppstein01, HSV-CGTA99}. Thus in
  the drawing of $G$, the subgraph induced by each colour class is
  crossing-free. The above theorems also produce drawings in which the
  edges are partitioned into non-crossing subgraph, but with
  additional proximity properties. Moreover, each subgraph of the
  partition is connected, which intuitively at least, is a desirable
  property in visualisation applications.

\item All our proofs are constructive, and lead to polynomial time
  algorithms (in the real RAM model). These algorithmic details are
  omitted.

\end{itemize}

\section{Relative Neighbourhood Graphs}
\seclabel{TreeDrawing}

To aid in the proofs of
\multithmref{Degree5Partition-MST}{TwoCovering-MST}, we now introduce
some notation and a number of geometric objects. Let $x$ and $y$ be
points in the plane. Let $|xy|$ be the Euclidean distance between $x$
and $y$. Let $\CIRC(x,\delta)$ be the circle of radius $\delta$
centred at $x$. Let $\DISC(x,\delta)$ be the open disc of radius
$\delta$ centred at $x$. Let $\CDISC(x,\delta)$ be the closed disc of
radius $\delta$ centred at $x$. As illustrated in \figref{Moon}, for
every real number $\delta$ such that $0<\delta<|xy|$,
let $$\LUNE(x,y,\delta):=(\DISC(y,\delta)-\CDISC(x,|xy|))\cup\{y\}\enspace.$$
The \emph{relative neighbourhood lens}\footnote{Unfortunately the
  computational geometry literature, and especially the literature on
  relative neighbourhood graphs, often refers incorrectly to a `lens' as a
  `lune'.} of $x$ and $y$
is $$\LENS(x,y):=\DISC(x,|xy|)\cap\DISC(y,|xy|)\enspace.$$

\Figure{Moon}{\includegraphics{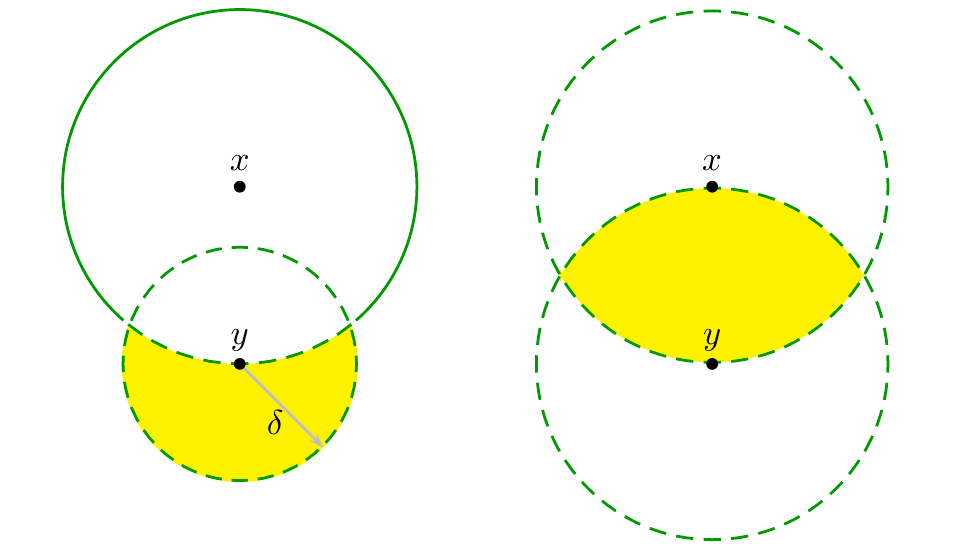}}{The regions $\LUNE(x,y,\delta)$
  and $\LENS(x,y)$}

Let $P\subset\RR$ be a finite set of points in the
plane. \citet{Toussaint-PR80} defined the \emph{relative neighbourhood
  graph} of $P$, denoted by $\RNG(P)$, to be the graph with vertex set
$P$, where two vertices $v,w\in P$ are adjacent if and only if
$\LENS(x,y)\cap P=\emptyset$. That is $v$ and $w$ are adjacent
whenever no vertex is simultaneously closer to $v$ than $w$ and closer
to $w$ than $v$. \citet{Toussaint-PR80} proved that
$\MST(P)\subseteq\RNG(P)$. Hence if $\RNG(P)$ is a tree, then
$\RNG(P)=\MST(P)$. The result of \citet{MonmaSuri-DCG92} mentioned in
\secref{Intro} was strengthened by \citet{BLL-Algo96} as follows.

\begin{lemma}[\citet{BLL-Algo96}]
  \lemlabel{Degree5Tree-RNG} Every degree-$5$ tree has a drawing that
  is the relative neighbourhood graph of its vertex set.
\end{lemma}

For all of the theorems introduced in \secref{Intro}, we in fact prove
stronger results about relative neighbourhood graphs.

\section{Drawings Based on a Partition}

\thmref{Degree5Partition-MST} is implied by the following 
result, since a relative neighbourhood graph that is a tree is a
minimum spanning tree.

\begin{theorem}
  \thmlabel{Degree5Partition-RNG} Let $\{T_1,\dots,T_k\}$ be a
  partition of a tree $T$ into degree-$5$ subtrees. Then there is a
  drawing of $T$ in which each $T_i$ is drawn as the relative
  neighbourhood graph of its vertex set.
\end{theorem}

\begin{proof}
  Let $D$ be the maximum distance between any two vertices in $T$ (the
  \emph{diameter} of $T$). Let $Q$ be the complete $5$-ary tree of
  height $D$. That is, every non-leaf vertex in $Q$ has degree $5$,
  and for some vertex $r$, the distance between $r$ and every leaf
  equals $D$.

  By \lemref{Degree5Tree-RNG}, there is a drawing of $Q$ that is the
  relative neighbourhood graph of its vertex set. Since the vertices
  of $Q$ are in general position, for some $\varepsilon>0$, for all
  distinct vertices $x,y\in V(Q)$, the discs $\DISC(x,\varepsilon)$
  and $\DISC(y,\varepsilon)$ are disjoint, and if $P$ is a point set
  that contains exactly one point from each disc
  $\DISC(x,\varepsilon)$ (where $x\in V(Q)$), then
  $Q\cong\RNG(P)$. (Here $\DISC(x,\varepsilon)$ means the disc centred
  at the point where $x$ is drawn.)\

  Define a homomorphism\footnote{A \emph{homomorphism} from a graph
    $G$ to a graph $H$ is a function $f:V(G)\rightarrow V(H)$ such
    that if $vw\in E(G)$ then $f(v)f(w)\in E(H)$.}  $f$ from $T$ to
  $Q$ as follows.  Choose an arbitrary starting vertex $v$ of $T$, let
  $f(v)=r$, and recursively construct a function $f$ such that
  $f(v)f(w)$ is an edge of $Q$ for every edge $vw$ of $T$, and if
  $f(v)f(w)=f(v')f(w')$ for distinct edges $vw\in E(T_i)$ and $v'w'\in
  E(T_j)$, then $i\ne j$. That is, edges in the same subtree are
  mapped to distinct edges of $Q$. Hence for each subtree $T_i$ of
  $T$, no two vertices in $T_i$ are mapped to the same vertex in $Q$
  (otherwise the image of the path in $T_i$ between the two vertices
  would form a cycle in $Q$). Moreover, if $Q_i$ is the subgraph of
  $Q$ induced by $\{f(v):v\in V(T_i)\}$ then $Q_i\cong T_i$. Draw each
  vertex $v\in V(T)$ at a distinct point $\phi(v)\in
  \DISC(f(v),\varepsilon)$ so that $\{\phi(v):v\in V(T)\}$ is in
  general position. Thus $P_i:=\{\phi(v):v\in V(T_i)\}$ contains
  exactly one point from each disc $\DISC(x,\varepsilon)$ where $x\in
  V(Q_i)$. Hence $T_i\cong Q_i\cong\RNG(P_i)$ as desired.
\end{proof}

\thmref{Degree4Partition-MST} is implied by the following stronger
result.

\begin{theorem}
  \thmlabel{Degree4Partition-RNG} Let $\{T_1,\dots,T_k\}$ be a
  partition of a rooted tree $T$ into outdegree-$4$ subtrees. Then
  there is a non-crossing drawing of $T$ such that each $T_i$ is
  drawn as the relative neighbourhood graph of its vertex set.
\end{theorem}


\thmref{Degree4Partition-RNG} is proved by induction with the
following hypothesis. This proof method generalises that of
\citet{BLL-Algo96}.

\begin{lemma}
  \lemlabel{Degree4Partition-RNG-proof} Let $\{T_1,\dots,T_k\}$ be a
  partition of a rooted tree $T$ into outdegree-$4$ subtrees. Let $r$
  be the root of $T$. Let $p$ and $q$ be distinct points in the
  plane. Let $\delta$ be a real number with $0<\delta<|pq|$. Then
  there is a non-crossing drawing of $T$ contained in
  $\LUNE(p,q,\delta)$ such that:
  \begin{itemize}
  \item $r$, which is drawn at $q$, is in $\LENS(x,p)$ for every
    vertex $x$ of $T-r$, and
  \item for all $i\in\{1,\dots,k\}$, the subtree $T_i$ is drawn as the
    relative neighbourhood graph of its vertex set.
  \end{itemize}
\end{lemma}





\begin{proof}
  We proceed by induction on $|V(T)|$. The result is trivial if
  $|V(T)|=1$.  Now assume that $|V(T)|\geq2$.  Let $\delta'$ be a real
  number with $0<\delta'<\delta$.  The circular arc
  $A:=\CIRC(q,\delta')-\DISC(p,|pq|)$ has an angle (measured from $q$)
  greater than $\pi$. Thus, as illustrated in \figref{FourPoints},
  there are four points $s_1,s_2,s_3,s_4$ in the interior of $A$, such
  that the angle (measured from $q$) between distinct points $s_i$ and
  $s_j$ is greater than $\frac{\pi}{3}$, implying
  $|s_iq|=|s_jq|<|s_is_j|$ and $q\in\LENS(s_i,s_j)$, and
  $\LENS(q,s_i)\cap\{s_1,s_2,s_3,s_4\}=\emptyset$.

  \Figure{FourPoints}{\includegraphics{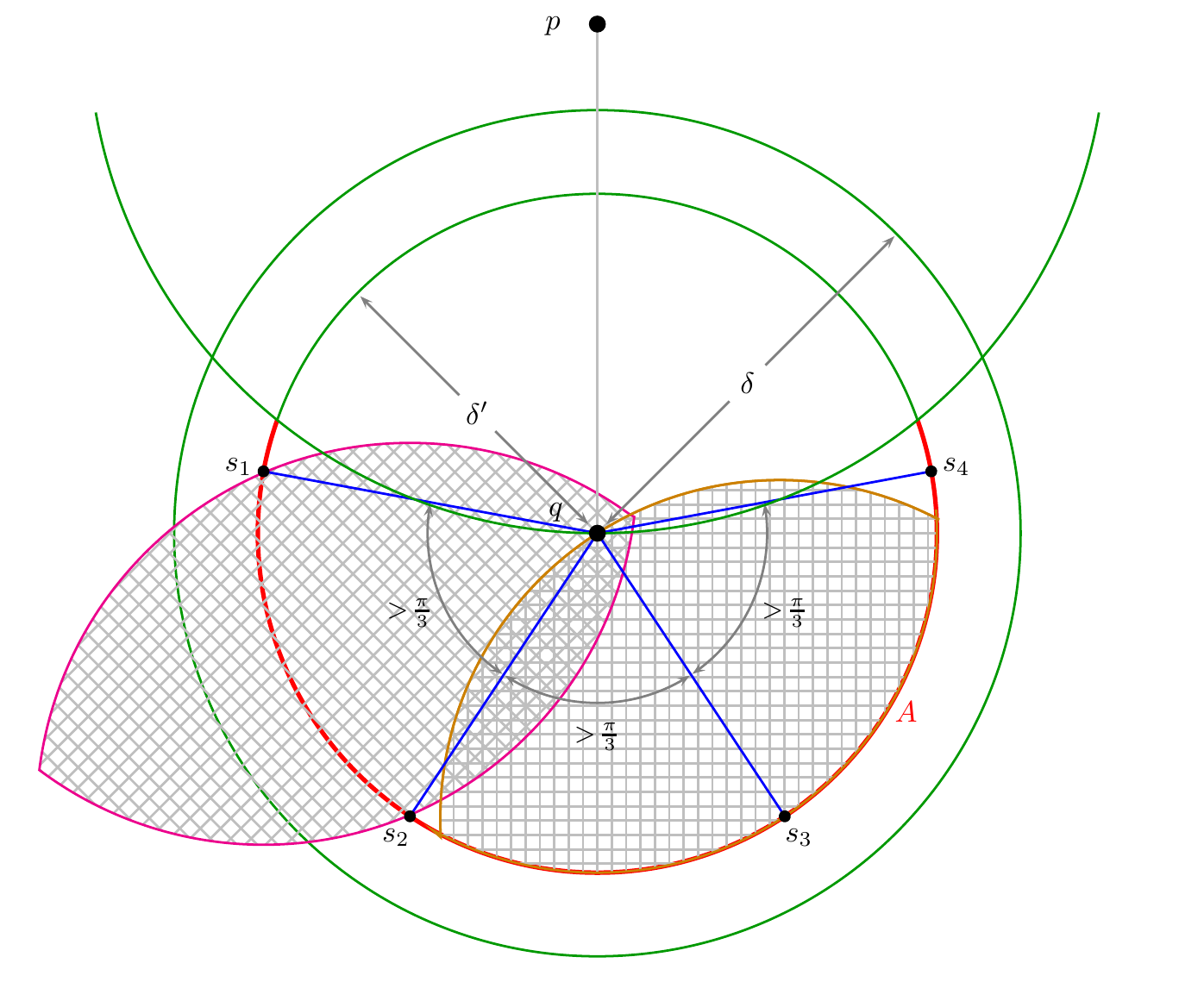}}{The points
    $s_1,s_2,s_3,s_4$, showing that $q\in\LENS(s_1,s_2)$ and
    $\LENS(q,s_3)\cap\{s_1,s_2,s_3,s_4\}=\emptyset$.}

  For small enough discs around the $s_i$, these properties are
  extended to every point in the disc. More precisely, there is a real
  number $\varepsilon\in(0,\delta')$ such that:
  \begin{enumerate}
  \item[(a)] $\DISC(s_i,\varepsilon)\subset\LUNE(p,q,\delta)$ for all
    $i\in\{1,2,3,4\}$;
  \item[(b)] $q\in\LENS(x,y)$ for all points
    $x\in\DISC(s_i,\varepsilon)$ and $y\in\DISC(s_j,\varepsilon)$ for
    all distinct $i, j\in\{1,2,3,4\}$;
  \item[(c)] $q\not\in\LENS(x,y)$ for all points
    $x,y\in\DISC(s_i,\varepsilon)$ for all $i\in\{1,2,3,4\}$; and
  \item[(d)] $\LENS(x,y)\cap \DISC(s_j,\varepsilon)=\emptyset$ for all
    points $x,y\in\DISC(s_i,\varepsilon)$ and for all distinct $i,
    j\in\{1,2,3,4\}$.
  \end{enumerate}

  For $j\in\{1,2,3,4\}$, since $\DISC(s_j,\varepsilon)$ has diameter
  $2\varepsilon$, there are points $t_{j,1},\dots,t_{j,k}$ on the arc
  $A\cap\DISC(s_j,\varepsilon)$ such that discs of radius
  $\frac{\varepsilon}{k}$ centred at $t_{j,1},\dots,t_{j,k}$ are
  pairwise disjoint, as illustrated in \figref{Degree4Partition}.

  \Figure{Degree4Partition}{\includegraphics{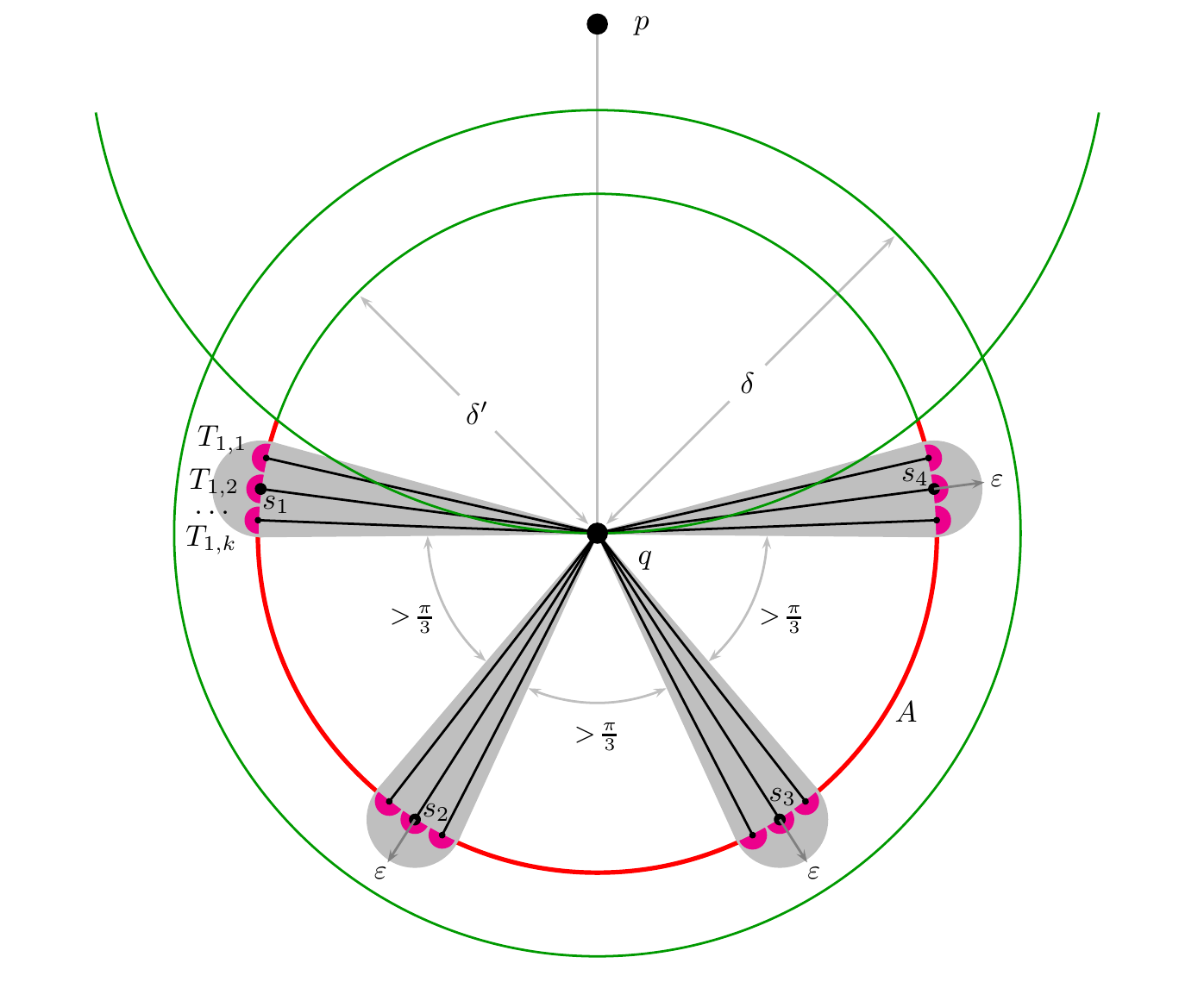}}{Construction
    in the proof of \lemref{Degree4Partition-RNG-proof}.}

  For $i\in\{1,\dots,k\}$, let $d_i$ be the outdegree of $r$ in
  $T_i$. So $d_i\in\{0,1,2,3,4\}$. Let $v_{i,1},\dots,v_{i,d_i}$ be
  the neighbours of $r$ in $T_i$. For $j\in\{1,\dots,d_i\}$, let
  $T_{i,j}$ be the component of $T-r$ that contains $v_{i,j}$. So
  $T_{i,j}$ is rooted at $v_{i,j}$, and $\{T_1\cap
  T_{i,j},\dots,T_k\cap T_{i,j}\}$ is a partition of $T_{i,j}$ into
  outdegree-$4$ subtreess. By induction, there is a non-crossing
  drawing of each $T_{i,j}$ contained in
  $\LUNE(q,t_{j,i},\frac{\varepsilon}{k})$ such that:
  \begin{enumerate}
  \item[(e)] $v_{i,j}$, which is drawn at $t_{j,i}$, is in
    $\LENS(x,q)$ for every vertex $x$ of $T_{i,j}-v_{i,j}$, and
  \item[(f)] for all $\ell\in\{1,\dots,k\}$, the subtree $T_{\ell}\cap
    T_{i,j}$ is drawn as the relative neighbourhood graph of its
    vertex set.
  \end{enumerate}

  Draw $r$ at $q$, and draw a straight-line edge from $r$ to each
  neighbour $v_{i,j}$ of $r$. Each subtree $T_{i,j}$ is drawn outside
  of $\DISC(q,\delta')$, while the edges incident to $r$ are contained
  within $\DISC(q,\delta')$, and therefore do not cross any other
  edge. Hence the drawing of $T$ is non-crossing.
  By (a), $T_{i,j}$ is drawn within
  $\LUNE(q,t_{j,i},\frac{\varepsilon}{k})\subset
  \DISC(t_j,\varepsilon)\subset \LUNE(p,q,\delta)$. The edges incident
  to $r$ are drawn within $\LUNE(p,q,\delta)$. Hence all of $T$ is
  drawn within $\LUNE(p,q,\delta)$.

  Now consider a vertex $x$ of $T-r$. Then $x$ is in $T_{i,j}$ for
  some $i\in\{1,2\dots,k\}$ and $j\in\{1,\dots,d_i\}$. Thus $x$ is
  drawn in $\DISC(q,\delta)-\DISC(p,|pq|)$, implying
  $|xq|<\delta<|xp|$ and $|pq|<|px|$. Hence $q\in\LENS(x,p)$, implying
  $r\in\LENS(x,p)$. This proves the first claim of the induction
  hypothesis.

  It remains to prove that each subtree $T_i$ is drawn as the relative
  neighbourhood graph of its vertex set. Consider distinct vertices
  $v$ and $w$ in $T_i$. We must show that $\LENS(v,w)\cap
  V(T_i)=\emptyset$ if and only if $vw\in E(T_i)$. Without loss of
  generality, $w\neq r$.

  \textbf{Case 1.} $v=r$ and $vw\in E(T_i)$: So $w=v_{i,j}$ for some
  $j\in\{1,2,3,4\}$. Then $v$ is drawn at $q$, and $w$ is drawn at
  $t_{j,i}$ . Now $\LENS(q,t_{j,i})\subset\DISC(q,\delta')$, which
  contains no vertex except $r$ (at $q$). Thus $\LENS(v,w)\cap
  V(T)=\emptyset$, as desired.

  \textbf{Case 2.} $v=r$ and $vw\not\in E(T_i)$: Then $w$ is in
  $T_{i,j}$ for some $j\in\{1,2,3,4\}$. Since $v$ is drawn at $q$, by
  induction, the vertex $t_{j,i}$, which is in $T_i$, is in
  $\LENS(v,w)$, as desired.

  Now assume that $v\neq r$ and $w\neq r$.

  \textbf{Case 3.} $v$ and $w$ are in the same component $T_{\ell,j}$
  of $T-r$: Then $v$ and $w$ are drawn within
  $\DISC(t_{\ell},\varepsilon)$. Each vertex in $T_i$ is $r$, is in
  $T_{\ell,j}$, or is in $T_{i,j'}$ for some $j'\neq j$.  Since $r$ is
  drawn at $q$, (c) implies that $r\not\in\LENS(v,w)$.  Since
  $T_{i,j'}$ is drawn within $\DISC(t_{j',i},\varepsilon)$, by (d),
  $\LENS(v,w)\cap V(T_{i,j'})=\emptyset$. Hence $\LENS(v,w)\cap
  V(T_i)=\emptyset$ if and only if $\LENS(v,w)\cap V(T_{\ell,j})\cap
  T_i=\emptyset$. By induction, $\LENS(v,w)\cap V(T_i)=\emptyset$ if
  and only if $v$ and $w$ are adjacent in $T_i$, as desired.



  \textbf{Case 4.} $v$ and $w$ are in distinct components of $T-r$:
  Thus $r$ is in $T_i$, $v$ is in $T_{i,j}$ and $w\in T_{i,j'}$ for
  some $j\neq j'$, and $v$ and $w$ are not adjacent. By construction,
  $v$ is drawn in $\DISC(s_j,\varepsilon)$ and $w$ is drawn in
  $\DISC(s_{j'},\varepsilon)$. Thus (b) implies that
  $q\in\LENS(v,w)$. Thus $r$, which is drawn at $q$, is in
  $\LENS(v,w)$, as desired.
\end{proof}







\section{Drawings with Large Angular Resolution}

\thmref{Degree3Partition-MST} is implied by the following stronger
result:

\begin{theorem}
  \thmlabel{Degree3Partition-RNG} Let $\{T_1,\dots,T_k\}$ be a
  partition of a rooted tree $T$ into outdegree-$3$ subtrees. Then
  there is a non-crossing drawing of $T$ with angular resolution at
  least $\frac{\pi}{\max\{\Delta^+(T)-1,4\}}$ such that each subtree
  $T_i$ is drawn as the relative neighbourhood graph of its vertex
  set.
\end{theorem}

\thmref{Degree3Partition-RNG} is proved by induction with the
following hypothesis.

\begin{lemma}
  \lemlabel{Degree3Partition-RNG-proof} Let $\{T_1,\dots,T_k\}$ be a
  partition of a rooted tree $T$ into outdegree-$3$ subtrees. Let $r$
  be the root of $T$. Let $p$ and $q$ be distinct points in the
  plane. Let $\delta$ be a real number with $0<\delta<|pq|$. Then
  there is a non-crossing drawing of $T$ contained in
  $\LUNE(p,q,\delta)$ such that:
  \begin{itemize}
  \item $r$, which is drawn at $q$, is in $\LENS(x,p)$ for every
    vertex $x$ of $T-r$, and
  \item for all $i\in\{1,\dots,k\}$, the subtree $T_i$ is drawn as the
    relative neighbourhood graph of its vertex set, and
  \item the drawing of $T$ has angular resolution greater than
    $\frac{\pi}{\max\{4,\Delta^+(T)-1\}}$.
  \end{itemize}
\end{lemma}

\begin{proof}
  We proceed by induction on $|V(T)|$. The result is trivial if
  $|V(T)|=1$.  Now assume that $|V(T)|\geq2$.  Let $\delta'$ be a real
  number with $0<\delta'<\delta$.

  Let $d:=\outdeg(r)$. For $i\in\{1,\dots,k\}$, let $d_i$ be the
  outdegree of $r$ in $T_i$. So $d_i\in\{0,1,2,3\}$ and
  $d=\sum_{i=1}^kd_i$. Let $v_{i,1},\dots,v_{i,d_i}$ be the neighbours
  of $r$ in $T_i$.  Let
$$X:=\{i:d_i=3\},\;\;\;
Y:=\{i:d_i=2\},\;\;\; Z:=\{i:d_i=1\}\enspace.$$ Thus
$d=3|X|+2|Y|+|Z|$. Partition $Z=Z'\cup Z''$ such that
$|Z''|\leq|Z'|\leq |Z''|+1$.

The circular arc $A:=\CIRC(q,\delta')-\DISC(p,|pq|)$ has an angle
(measured from $q$) greater than $\pi$. Thus there are points
$s_1,\dots,s_d$ in this order on $A$ such that the angle (measured
from $q$) between distinct points $s_a$ and $s_b$ is greater than
$\frac{\pi |b-a|}{d-1}$.

Let $\preceq$ be the total ordering of the neighbours of $r$ such that
$ \{v_{i,1}:i\in X\}\preceq\{v_{i,1}:i\in Y\} \preceq\{v_{i,1}:i\in
Z'\}\preceq\{v_{i,2}:i\in X\}\preceq\{v_{i,2}:i\in
Y\}\preceq\{v_{i,1}:i\in Z''\}\preceq\{v_{i,3}:i\in X\}$, where within
each set, the vertices are ordered by their $i$-value.  Draw the
neighbours of $r$ in the order of $\preceq$ at $s_1,\dots,s_d$. That
is, the first vertex in $\preceq$ is drawn at $s_1$, the second vertex
in $\preceq$ is drawn at $s_2$, and so on. Let $t_{i,j}$ be the point
where $v_{i,j}$ is drawn.

Consider distinct vertices $v_{i,j}$ and $v_{i,\ell}$ in some subtree
$T_i$ such that $\ell>j$. Say $t_{i,j}=s_a$ and
$t_{i,\ell}=s_b$. Observe that $b-a\geq |X|+|Y|+|Z''|\geq
|X|+|Y|+\half(|Z|-1)\geq \third(3|X|+2|Y|+|Z|-1)=\frac{d-1}{3}$.
Hence the angle (measured from $q$) between $v_{i,j}$ and $v_{i,\ell}$
is greater than $\frac{\pi (d-1)/3}{d-1}=\frac{\pi}{3}$.  This implies
that $|t_{i,j}q|=|t_{i,\ell} q|<|t_{i,j}t_{i,\ell}|$. Thus
$q\in\LENS(t_{i,j},t_{i,\ell})$ and
$t_{i,\ell}\not\in\LENS(q,t_{i,j})$ and
$t_{i,j}\not\in\LENS(q,t_{i,\ell})$.

For small enough discs around $s_1,\dots,s_d$, these properties are
extended to every point in the disc. More precisely, there is a real
number $\varepsilon\in(0,\delta')$ such that:
\begin{enumerate}
\item[(a)] $\DISC(s_a,\varepsilon)\subset\LUNE(p,q,\delta)$ for all
  $a\in\{1,\dots,d\}$;
\item[(b)] $q\in\LENS(x,y)$ for all points
  $x\in\DISC(t_{i,j},\varepsilon)$ and
  $y\in\DISC(t_{i,\ell},\varepsilon)$ for all distinct vertices
  $v_{i,j}$ and $v_{i,\ell}$ in the same subtree $T_i$;
\item[(c)] $q\not\in\LENS(x,y)$ for all points
  $x,y\in\DISC(s_a,\varepsilon)$ for all $a\in\{1,\dots,d\}$; and
\item[(d)] $\LENS(x,y)\cap \DISC(s_b,\varepsilon)=\emptyset$ for all
  distinct $a,b\in\{1,\dots,d\}$ and for all points
  $x,y\in\DISC(s_a,\varepsilon)$.
\end{enumerate}

\Figure{OutdegreeThree}{\includegraphics{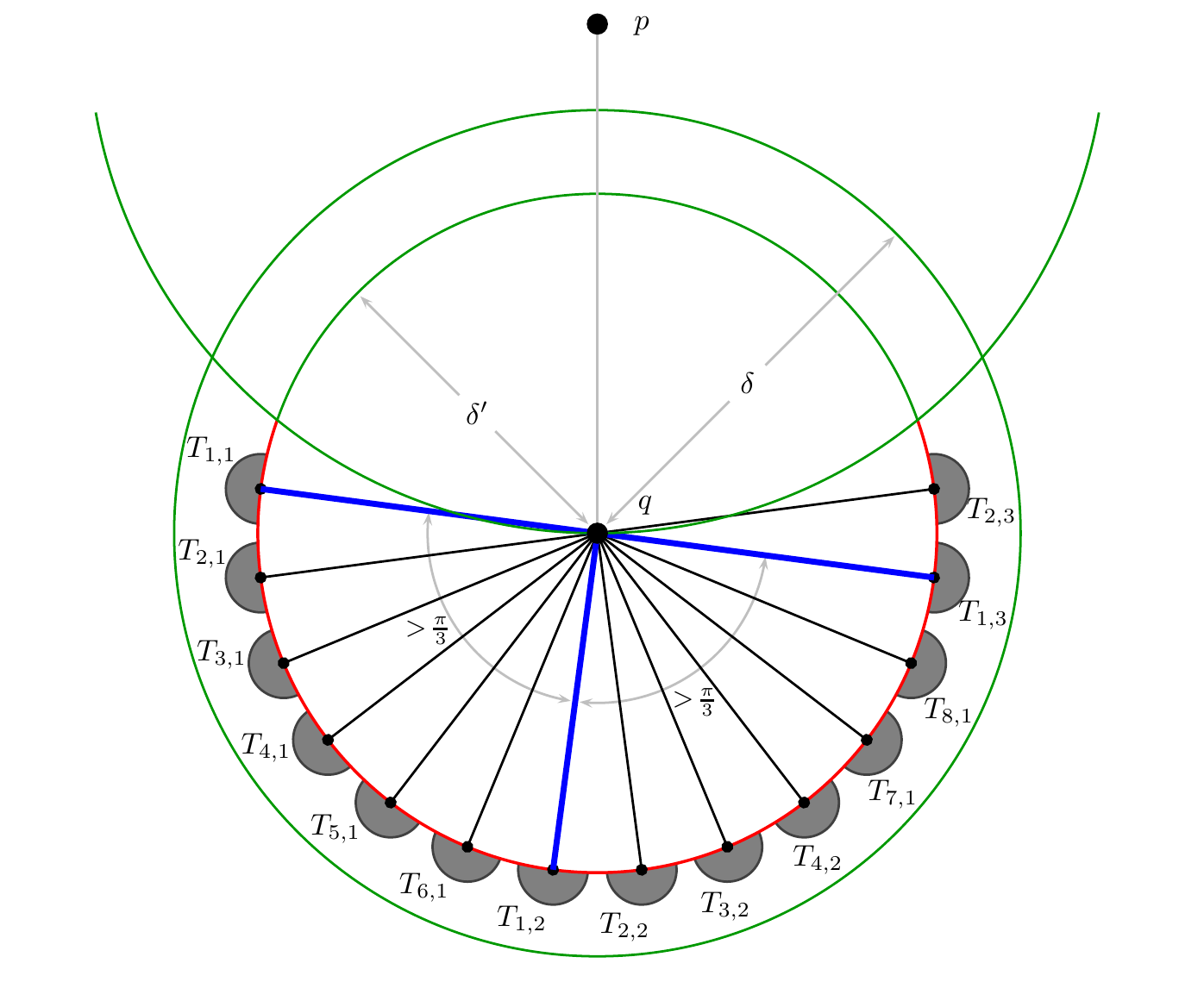}}{Construction
  in the proof of \lemref{Degree3Partition-RNG-proof}. Here
  $X=\{1,2\}$, $Y=\{3,4\}$, $Z'=\{5,6\}$ and $Z''=\{7,8\}$. The tree
  $T_1$ is highlighted.}

For $i\in\{1,\dots,k\}$ and $j\in\{1,\dots,d_i\}$, let $T_{i,j}$ be
the component of $T-r$ that contains $v_{i,j}$.  Each subtree
$T_{i,j}$ is rooted at $v_{i,j}$, and $\{T_1\cap T_{i,j},\dots,T_k\cap
T_{i,j}\}$ is a partition of $T_{i,j}$ into outdegree-$3$
subtreess. By induction, there is a non-crossing drawing of $T_{i,j}$
contained in $\LUNE(q,t_{i,j},\varepsilon)$ such that:
\begin{enumerate}
\item[(e)] $v_{i,j}$, which is drawn at $t_{i,j}$, is in $\LENS(x,q)$
  for every vertex $x$ of $T_{i,j}-v_{i,j}$; and
\item[(f)] for all $\ell\in\{1,\dots,k\}$, the subtree $T_{\ell}\cap
  T_{i,j}$ is drawn as the relative neighbourhood graph of its vertex
  set; and
\item[(g)] the drawing of $T_{i,j}$ has angular resolution greater
  than $\frac{\pi}{\max\{\Delta^+(T_{i,j})-1,4\}}$, which is at least
  $\frac{\pi}{\max\{\Delta^+(T)-1,4\}}$.
\end{enumerate}

Draw $r$ at $q$, and draw a straight-line edge from $r$ to each
neighbour $v_{i,j}$ of $r$. The angle between two edges incident to
$r$ is at least $\frac{\pi}{d-1}\geq \frac{\pi}{\Delta^+(T)-1}$. The
angle between an edge $rv_{i,j}$ and each edge $v_{i,j}x$ in $T_{i,j}$
is at least $\frac{\pi}{4}$. With (g), this proves the third claim of
the lemma.

Each subtree $T_{i,j}$ is drawn outside of $\DISC(q,\delta')$, while
the edges incident to $r$ are contained within $\DISC(q,\delta')$, and
therefore do not cross any other edge. Hence the drawing of $T$ is
non-crossing.  By (a), $T_{i,j}$ is drawn within
$\LUNE(q,t_{i,j},\varepsilon)\subset\DISC(t_{i,j},\varepsilon)\subset
\LUNE(p,q,\delta)$. The edges incident to $r$ are drawn within
$\LUNE(p,q,\delta)$. Hence all of $T$ is drawn within
$\LUNE(p,q,\delta)$.

Now consider a vertex $x$ of $T-r$. Then $x$ is in $T_{i,j}$ for some
$i\in\{1,\dots,k\}$ and $j\in\{1,\dots,d_i\}$. Thus $x$ is drawn in
$\DISC(q,\delta)-\DISC(p,|pq|)$, implying $|xq|<\delta<|xp|$ and
$|pq|<|px|$. Hence $q\in\LENS(x,p)$, implying $r\in\LENS(x,p)$. This
proves the first claim of the lemma.

It remains to prove that each subtree $T_i$ is drawn as the relative
neighbourhood graph of its vertex set. Consider distinct vertices $v$
and $w$ in $T_i$. We must show that $\LENS(v,w)\cap V(T_i)=\emptyset$
if and only if $vw\in E(T_i)$. Without loss of generality, $w\neq r$.

\textbf{Case 1.} $v=r$ and $vw\in E(T_i)$: So $w=v_{i,j}$ for some
$j\in\{1,2,3\}$. Then $v$ is drawn at $q$, and $w$ is drawn at
$t_{i,j}$ . Now $\LENS(q,t_{i,j})\subset\DISC(q,\delta')$, which
contains no vertex except $r$ (at $q$). Thus $\LENS(v,w)\cap
V(T)=\emptyset$, as desired.

\textbf{Case 2.} $v=r$ and $vw\not\in E(T_i)$: Then $w$ is in
$T_{i,j}$ for some $j\in\{1,2,3\}$, but $w\neq v_{i,j}$. Since $v$ is
drawn at $q$, by (e), the vertex $v_{i,j}$, which is in $T_i$, is in
$\LENS(v,w)$, as desired.

Now assume that $v\neq r$ and $w\neq r$.

\textbf{Case 3.} $v$ and $w$ are in the same component $T_{\ell,j}$ of
$T-r$, for some $\ell\in\{1,\dots,k\}$: Then $v$ and $w$ are drawn
within $\DISC(t_{\ell,j},\varepsilon)$. Each vertex in $T_i$ is $r$,
is in $T_{\ell,j}$, or is in $T_{i,j'}$ for some $(i,j')\neq(\ell,j)$.
Since $r$ is drawn at $q$, (c) implies that $r\not\in\LENS(v,w)$.
Since $T_{i,j'}$ is drawn within $\DISC(t_{i,j'},\varepsilon)$, by
(d), $\LENS(v,w)\cap V(T_{i,j'})=\emptyset$. Hence $\LENS(v,w)\cap
V(T_i)=\emptyset$ if and only if $\LENS(v,w)\cap V(T_{\ell,j})\cap
T_i=\emptyset$. By (f), $\LENS(v,w)\cap V(T_i)=\emptyset$ if and only
if $v$ and $w$ are adjacent in $T_i$, as desired.



\textbf{Case 4.} $v$ and $w$ are in distinct components of $T-r$: Thus
$r$ is in $T_i$, $v$ is in $T_{i,j}$ and $w\in T_{i,j'}$ for some
$j\neq j'$, and $v$ and $w$ are not adjacent. By construction, $v$ is
drawn in $\DISC(t_{i,j},\varepsilon)$ and $w$ is drawn in
$\DISC(t_{i,j'},\varepsilon)$. Thus (b) implies that
$q\in\LENS(v,w)$. Thus $r$, which is drawn at $q$, is in $\LENS(v,w)$,
as desired.

Therefore the subtree $T_i$ is drawn as the relative neighbourhood
graph of its vertex set. This completes the proof.
\end{proof}

\section{Drawings Based on a Covering}
\seclabel{Covering}

\thmref{TwoCovering-RNG} below establishes a result for relative
neighbourhood graphs that implies \thmref{TwoCovering-MST} for minimum
spanning trees. Before proving \thmref{TwoCovering-RNG} we give a
simpler proof of a weaker result, in which the obtained drawing might
have crossings.

\begin{proposition}
  \proplabel{TwoCovering-RNG-Crossings} Let $\{T_1,T_2\}$ be a
  covering of a tree $T$ by degree-$5$ subtrees. Then there is a
  drawing of $T$ in which each $T_i$ is drawn as the relative
  neighbourhood graph of its vertex set.
\end{proposition}

\begin{proof}
  We proceed by induction on $|V(T)|$. If $\Delta(T)\leq 5$ then
  $T\cong\RNG(P)$ for some point set $P$ by
  \lemref{Degree5Tree-RNG}. This drawing is crossing-free since it
  also a minimum spanning tree. Furthermore, each $T_i$ is drawn as
  the relative neighbourhood graph of the subset of $P$ representing
  $T_i$. Now assume that $\Delta(T)\geq6$. Thus $\deg_T(v)\geq6$ for
  some vertex $v$. Hence there are edges $vx\in E(T_1)-E(T_2)$ and
  $vy\in E(T_2)-E(T_1)$. Let $T'$ be the tree obtained from $T$ by
  identifying $x$ and $y$ into a new vertex $w$. (This operation is
  called an \emph{elementary homomorphism} or \emph{folding}; see
  \citep{CookEvans, Evans-JGT86, GaoHahn-DM95, BS-CN81} and
  \figref{FoldTree}.)\ Let $T'_i$ be the subtrees of $T'$ determined
  by $T_i$ for $i\in\{1,2\}$. Note that the edge $vw$ is in $T_1'\cap
  T_2'$. Observe that $\{T_1',T_2'\}$ is a covering of $T'$ by
  degree-$5$ subtrees. By induction, there is a drawing of $T'$ such
  that each $T_i'$ is the relative neighbourhood graph of its vertex
  set. Moreover, for some $\varepsilon>0$, if $w$ is moved to any
  point in $\DISC(w,\varepsilon)$ then in the resulting drawing of
  $T'$, each $T_i'$ is drawn as the relative neighbourhood graph of
  its vertex set. Consider a drawing of $T$ in which every vertex in
  $V(T)-\{x,y\}$ inherits is position in the drawing of $T'$, and $x$
  and $y$ are assigned distinct points in
  $\DISC(w,\varepsilon)$. Since $x\in V(T_1)-V(T_2)$ and $y\in
  V(T_2)-V(T_1)$, each $T_i$ is drawn as the relative neighbourhood
  graph of its vertex set in the drawing of $T$.
\end{proof}

\Figure{FoldTree}{\includegraphics{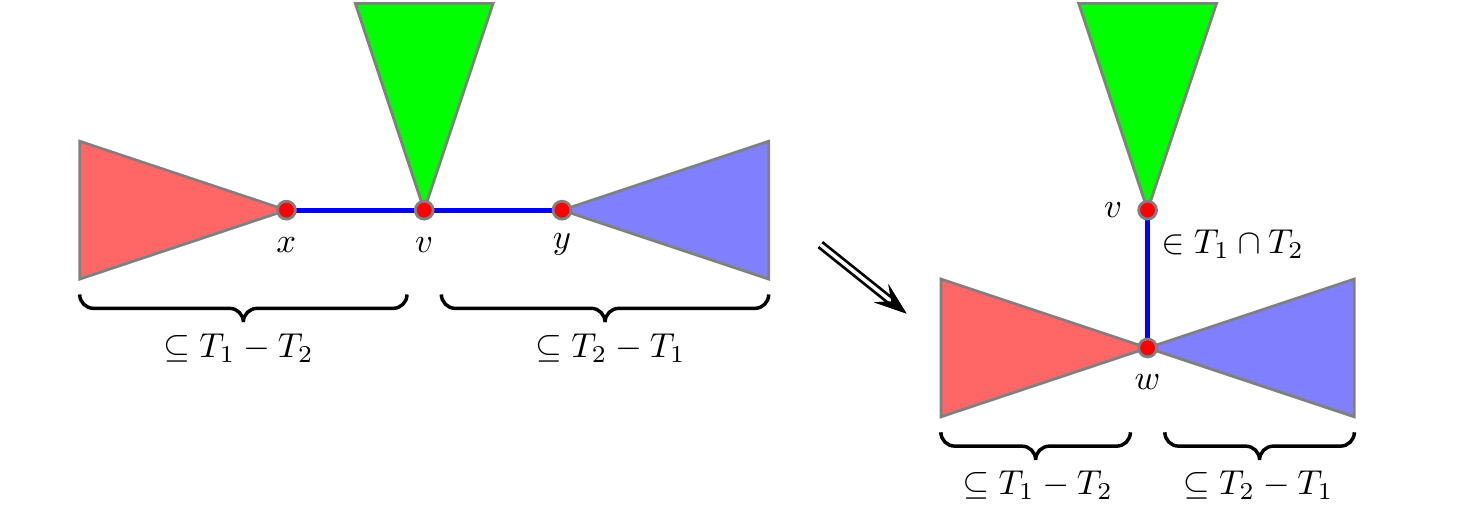}}{Folding the tree $T$ in
  the proof of \propref{TwoCovering-RNG-Crossings}.}

We now strengthen \propref{TwoCovering-RNG-Crossings} by showing that
the drawing of $T$ can be made crossing-free. \thmref{TwoCovering-MST}
is implied by the following stronger result:

\begin{theorem}
  \thmlabel{TwoCovering-RNG} Let $\{T_1,T_2\}$ be a covering of a tree
  $T$ by degree-$5$ subtrees. Then there is a non-crossing drawing of
  $T$ such that each $T_i$ is drawn as the relative neighbourhood
  graph of its vertex set.
\end{theorem}

The proof of \thmref{TwoCovering-RNG} depends on the following
definition. A \emph{combinatorial embedding} of a graph is a cyclic
ordering of the edges incident to each vertex. We define a
combinatorial embedding of a graph $G$, with respect to a covering
$\{G_1,G_2\}$ of $G$, to be \emph{good} if for each vertex $v$ of $G$,
in the clockwise ordering of the edges incident to $v$, the edges in
$E(G_1)-E(G_2)$ are grouped together, followed by the edges in
$E(G_1)\cap E(G_2)$, followed by the edges in $E(G_2)-E(G_1)$. Since
every tree, covered by two subtrees, obviously has a good embedding,
\thmref{TwoCovering-RNG} now follows from the next lemma:

\begin{lemma}
  \lemlabel{TwoCovering-RNG} Let $\{T_1,T_2\}$ be a covering of a tree
  $T$ by degree-$5$ subtrees. For every good combinatorial embeddding
  of $T$, with respect to $\{T_1,T_2\}$, there is a non-crossing
  drawing of $T$ such that each $T_i$ is drawn as the relative
  neighbourhood graph of its vertex set, and the given combinatorial
  embedding of $T$ is preserved in the drawing.
\end{lemma}

\begin{proof}
  We proceed by induction on $|V(T)|$. If $\Delta(T)\leq 5$ then
  $T\cong\RNG(P)$ for some point set $P$ by
  \lemref{Degree5Tree-RNG}. This drawing is crossing-free since it
  also a minimum spanning tree. Moreover, by examining the proof of
  \lemref{Degree5Tree-RNG}, it is easily seen that any given
  combinatorial embedding of $T$ can be preserved in the drawing. Each
  $T_i$ is drawn as the relative neighbourhood graph of the subset of
  $P$ representing $T_i$. Now assume that $\deg_T(v)\geq6$ for some
  vertex $v$. Hence there are edges $vx\in E(T_1)-E(T_2)$ and $vy\in
  E(T_2)-E(T_1)$ such that $vx$ and $vy$ are consecutive in the cyclic
  ordering of the edges incident to $v$.

  Let $T'$ be the tree obtained from $T$ by identifying $x$ and $y$
  into a new vertex $w$. Let $T'_i$ be the subtrees of $T'$ determined
  by $T_i$ for $i\in\{1,2\}$. Note that the edge $vw$ is in $T_1'\cap
  T_2'$. The cyclic ordering of the edges in $T'$ incident to $v$ is
  obtained from the cyclic ordering of the edges in $T$ incident to
  $v$ by replacing $vx$ and $vy$ (which are consecutive) by $vw$. And
  $N_{T'}(w)$ is ordered
  $(N_{T_1-E(T_2)}(x),wv,N_{T_2-E(T_1)}(y))$. Other vertices keep
  their ordering in $T$.

  Observe that $\{T_1',T_2'\}$ is a covering of $T'$ by degree-$5$
  subtrees. By induction, there is a non-crossing drawing of $T'$ such
  that each $T_i'$ is the relative neighbourhood graph of its vertex
  set, and the given combinatorial embedding of $T$ is preserved in
  the drawing. For some $\varepsilon>0$, if $w$ is moved to any point
  in $\DISC(w,\varepsilon)$ then in the resulting drawing of $T'$,
  each $T_i'$ is drawn as the relative neighbourhood graph of its
  vertex set, and the given combinatorial embedding of $T$ is
  preserved. Consider a drawing of $T$ in which every vertex in
  $V(T)-\{x,y\}$ inherits is position in the drawing of $T'$, and $x$
  and $y$ are assigned distinct points in
  $\DISC(w,\varepsilon)$. Since $x\in V(T_1)-V(T_2)$ and $y\in
  V(T_2)-V(T_1)$, each $T_i$ is drawn as the relative neighbourhood
  graph of its vertex set in the drawing of $T$. It remains to assign
  points for $x$ and $y$ in $\DISC(w,\varepsilon)$ so that the drawing
  of $T$ is crossing-free. In the drawing of $T'$, the edges incident
  to $w$ are ordered $(N_{T_1-E(T_2)}(x),wv,N_{T_2-E(T_1)}(y))$. Let
  $R$ be a ray centred at $w$ that separates the edges in $T_1-E(T_2)$
  incident to $w$ and those in $T_2-E(T_1)$ incident to $w$, such that
  $v$ is not on the extension of $R$. At most one of $x$ and $y$, say
  $x$, has neighbours on both sides of the extension of $R$.  As
  illustrated in \figref{UnFoldTree}, position $x$ at $w$, and
  position $y$ on $R$ and inside $\DISC(w,\varepsilon)$. It follows
  that there are no crossings and the correct ordering of edges is
  preserved at $v$, $x$ and $y$.
\end{proof}

\Figure{UnFoldTree}{\includegraphics{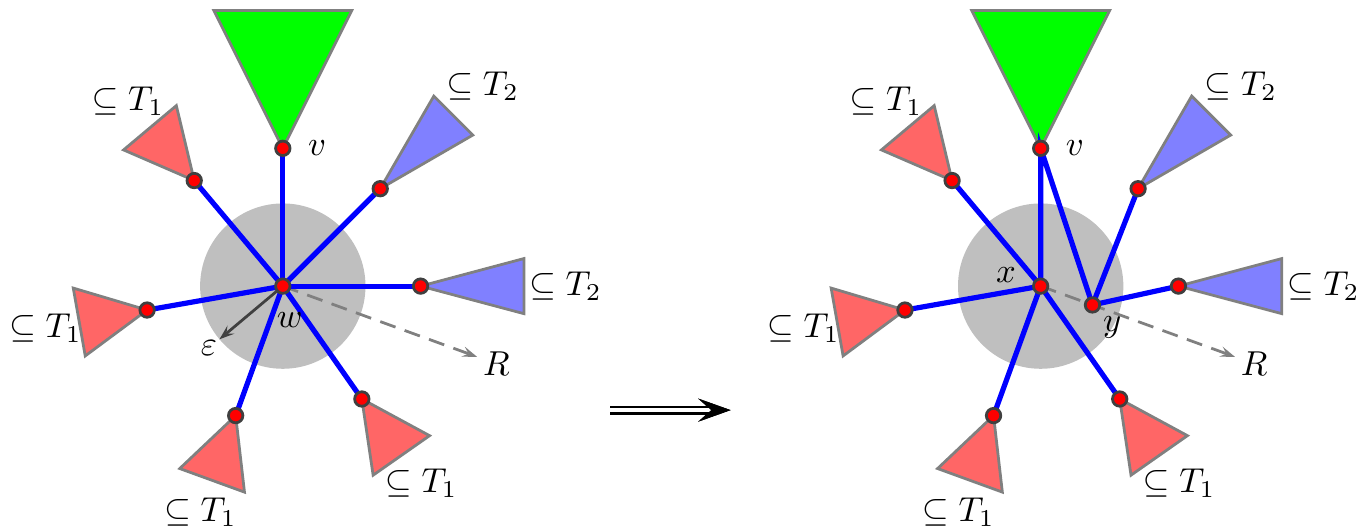}}{Producing a drawing
  of $T$ given a drawing of $T'$ in the proof of
  \lemref{TwoCovering-RNG}.}

We now show that \thmref{TwoCovering-MST} cannot be generalised for
coverings by three or more subtrees. (Thus neither
\propref{TwoCovering-RNG-Crossings} nor \thmref{TwoCovering-RNG} can
be similarly generalised.)\ Let $T$ be the $6$-star with root $r$ and
leaves $v_1,\dots,v_6$.  Let $\{T_1,T_2,T_3\}$ be the following
covering of $T$.  Let $T_1$ be the subtree of $T$ induced by
$\{r,v_1,v_2,v_3,v_4\}$.  Let $T_2$ be the subtree of $T$ induced by
$\{r,v_1,v_2,v_5,v_6\}$.  Let $T_3$ be the subtree of $T$ induced by
$\{r,v_3,v_4,v_5,v_6\}$.  Thus each $T_i$ is a $4$-star.  Suppose on
the contrary that $T$ has a drawing such that each $T_i$ is drawn as a
minimum spanning tree of its vertex set.  The angle $\angle v_irv_j$
between some pair of consecutive edges $rv_i$ and $rv_j$ (in the
cyclic order around $r$) is less than $\frac{\pi}{3}$ since no three
vertices are collinear.  Since $v_i$ and $v_j$ are each in two
subtrees, and $r$ is in every subtree, the vertices $r,v_i,v_j$ are in
a common subtree $T_{\ell}$.  Every minimum spanning tree has angular
resolution at least $\frac{\pi}{3}$. Thus $T_{\ell}$ is not drawn as a
minimum spanning tree. This contradiction proves there is no drawing
of $T$ such that each $T_i$ is drawn as a minimum spanning tree of its
vertex set.  Note that this argument generalises to show that if
$P_1,\dots,P_{15}$ are the $\binom{6}{2}$ paths through the root of
the $6$-star $T$, then in every drawing of $T$, some $P_i$ is not a
minimum spanning tree of its vertex set.

\section{Further Research}

This paper has not analysed the area of the drawings produced by our
algorithms.  It would be interesting to consider whether there are
drawings whose area is polynomial in the number of vertices of the
given tree, for example when the tree is partitioned into outdegree-3
subtrees.  While the problem of drawing a tree as a minimum spanning
tree in polynomial area is open in the general case
\citep{MonmaSuri-DCG92}, \citet{Kaufmann-GD05} proved that every
degree-4 tree has a drawing as a minimum spanning tree in polynomial
area; also see \citep{PV-CGTA04}.



A second direction for further research is to extend the approach used
in this paper to other types of proximity drawings of trees; see
\citep{Liotta-ProximityDrawings}.  For example, every degree-4 tree
admits a w-$\beta$-drawing for all values of $\beta$ in 
$(\cos(\frac{2 \pi}{5})^{-1}, \infty)$; see
\citep[Theorem~7]{DLW-JDA06}. Given a partition of a rooted tree $T$
into outdegree-3 subtrees and a value of $\beta$ in the above
interval, is there a drawing of $T$ in which each subtree is drawn as
a w-$\beta$-drawing?

The results of this paper motivate studying coverings and partitions
of trees by subtrees of bounded degree. We consider these purely
combinatorial problems in our companion paper
\citep{Wood-Coverings}. For example, given a tree $T$ and integer $d$,
we present there a formula for the minimum number of degree-$d$
subtrees that partition $T$, and describe a polynomial time algorithm
that finds such a partition. Similarly, we present a polynomial time
algorithm that finds a covering of $T$ by the minimum number of
degree-$d$ subtrees.


\def\cprime{$'$} \def\soft#1{\leavevmode\setbox0=\hbox{h}\dimen7=\ht0\advance
  \dimen7 by-1ex\relax\if t#1\relax\rlap{\raise.6\dimen7
  \hbox{\kern.3ex\char'47}}#1\relax\else\if T#1\relax
  \rlap{\raise.5\dimen7\hbox{\kern1.3ex\char'47}}#1\relax \else\if
  d#1\relax\rlap{\raise.5\dimen7\hbox{\kern.9ex \char'47}}#1\relax\else\if
  D#1\relax\rlap{\raise.5\dimen7 \hbox{\kern1.4ex\char'47}}#1\relax\else\if
  l#1\relax \rlap{\raise.5\dimen7\hbox{\kern.4ex\char'47}}#1\relax \else\if
  L#1\relax\rlap{\raise.5\dimen7\hbox{\kern.7ex
  \char'47}}#1\relax\else\message{accent \string\soft \space #1 not
  defined!}#1\relax\fi\fi\fi\fi\fi\fi} \def\Dbar{\leavevmode\lower.6ex\hbox to
  0pt{\hskip-.23ex\accent"16\hss}D}

\end{document}